%% file: main.tex
\documentclass[runningheads]{llncs}
\usepackage[utf8]{inputenc}
\usepackage{cite}
\usepackage{amsmath,amssymb,amsfonts}
\usepackage{algpseudocode}
\usepackage{graphicx}
\usepackage{pseudo}
\usepackage{textcomp}
\usepackage{xcolor}
\usepackage{enumitem}
\usepackage{theoremref}
\usepackage{todonotes}
\usepackage{comment}
\usepackage[colorlinks=true,linkcolor={blue},citecolor=blue]{hyperref}
\usepackage[noabbrev,capitalise,nameinlink]{cleveref}
\usepackage{microtype}

\newcommand{\WO}{\textsf{W[1]}}
\newcommand{\XL}{\textsf{XL}}
\newcommand{\PP}{\textsf{P}}
\newcommand{\PPSPACE}{\textsf{PSPACE}}
\newcommand{\NPP}{\textsf{NP}}
\newcommand{\FPT}{\textsf{FPT}}

\title{Token sliding on graphs of girth five\thanks{This work is supported by PHC Cedre project 2022 ``PLR''. A Preliminary version of the work appeared in the 48th International Workshop on Graph-Theoretic Concepts in Computer Science (WG 2022).}}
\author{Valentin Bartier\inst{1}\thanks{Supported by ANR project GrR (ANR-18-CE40-0032).} \and Nicolas Bousquet\inst{1}\thanks{Supported by ANR project GrR (ANR-18-CE40-0032).} \and Jihad Hanna\inst{2} \and Amer~E.~Mouawad\inst{2,3}\thanks{Research supported by
the Alexander von Humboldt Foundation and partially supported by URB project ``A theory of change through the lens of reconfiguration''.} \and Sebastian Siebertz\inst{3}}
\authorrunning{N.~Bousquet, J.~Hanna, A.E.~Mouawad, and S.~Siebertz}

\institute{CNRS, LIRIS, Universit\'e de Lyon, Universit\'e Claude Bernard Lyon 1, France \and American University of Beirut, Lebanon \and University of Bremen, Germany}

\begin{document}
\maketitle

\input{abstract}
\input{intro}
\input{prelim}
\input{reducing}
\input{algorithm}

\bibliographystyle{plain}
\bibliography{refs}

\end{document}

%% file: abstract.tex
\begin{abstract}
In the {\sc Token Sliding} problem we are given a graph $G$ and two independent sets $I_s$ and $I_t$ in $G$ of size $k \geq 1$. The goal is to decide whether there exists a sequence $\langle I_1, I_2, \ldots, 
I_\ell \rangle$ of independent sets such that for all $i \in \{1,\ldots, \ell\}$ the set $I_i$ is an independent set of size~$k$, $I_1 = I_s$, $I_\ell = I_t$ and $I_i \triangle I_{i + 1} = \{u, v\} \in E(G)$. Intuitively, we view each independent set as a collection of tokens placed on the vertices of the graph. Then, the problem asks  whether there exists a sequence of independent sets that transforms $I_s$ into $I_t$ where at each step we are allowed to slide one token from a vertex to a neighboring vertex. In this paper, we focus on the parameterized complexity of {\sc Token Sliding} parameterized by $k$. As shown by Bartier et al.~\cite{DBLP:journals/algorithmica/BartierBDLM21}, the problem is \WO-hard on graphs of girth four or less, and the authors posed the question of whether there exists a constant $p \geq 5$ such that the problem becomes fixed-parameter tractable on graphs of girth at least $p$. We answer their question positively and prove that the problem is indeed fixed-parameter tractable on graphs of girth five or more, which establishes a full classification of the tractability of {\sc Token Sliding} parameterized by the number of tokens based on the girth of the input graph. 
\end{abstract}

%% file: intro.tex
\section{Introduction}
Many algorithmic questions present themselves in the following form: Given the description of a system state and the description of a state we would prefer the system to be in, is it possible to transform the system from its current state into the more desired one without ``breaking'' certain properties of the system in the process? Such questions, with some generalizations and specializations, have
received a substantial amount of attention under the so-called \emph{combinatorial reconfiguration framework}~\cite{DBLP:journals/tcs/BrewsterMMN16,DBLP:books/cu/p/Heuvel13,DBLP:journals/siamdm/Wrochna20}.

Historically, the study of reconfiguration questions predates the field of computer science, as many classic one-player games can be formulated as reachability questions~\cite{DBLP:journals/JS79,DBLP:journals/icga/KendallPS08}, e.g., the $15$-puzzle and Rubik's cube. More recently, reconfiguration problems have emerged from computational problems in different areas such as graph theory~\cite{DBLP:journals/dm/CerecedaHJ08,DBLP:journals/tcs/ItoDHPSUU11,DBLP:journals/dam/ItoKD12}, constraint satisfaction~\cite{DBLP:journals/siamcomp/GopalanKMP09,DBLP:journals/siamdm/MouawadNPR17}, computational geometry~\cite{DBLP:journals/comgeo/LubiwP15}, and even quantum complexity theory~\cite{DBLP:journals/toct/GharibianS18}. We refer the reader to the surveys by van den Heuvel~\cite{DBLP:books/cu/p/Heuvel13} and Nishimura~\cite{DBLP:journals/algorithms/Nishimura18} for extensive background on combinatorial reconfiguration.

\paragraph*{Independent set reconfiguration.} In this work, we focus on the reconfiguration of independent sets. Given a simple undirected graph $G$, a set of vertices $S \subseteq V(G)$ is an \emph{independent set} if the vertices of this set are pairwise non-adjacent. Finding an independent set of size $k$, i.e., the {\sc Independent Set} problem, is known to be \NPP-hard, but also \WO-hard\footnote{Informally, this means that it is unlikely to be fixed-parameter tractable.} parameterized by solution size $k$ and not approximable within $O(n^{1-\epsilon})$, for any $\epsilon > 0$, unless $\PP = \NPP$ \cite{DBLP:journals/toc/Zuckerman07}. Moreover, {\sc Independent Set} remains \WO-hard on graphs excluding $C_4$ (the cycle on four vertices) as an induced subgraph \cite{DBLP:conf/iwpec/BonnetBCTW18}.

We view an independent set as a collection of tokens placed on the vertices of a graph such that no two tokens are placed on adjacent vertices. This gives rise to two natural adjacency relations between independent sets (or token configurations), also called \emph{reconfiguration steps}. These reconfiguration steps, in turn, give rise to two combinatorial reconfiguration problems.

In the {\sc Token Sliding} problem, introduced by Hearn and Demaine \cite{DBLP:journals/tcs/HearnD05}, two independent sets are adjacent if one can be obtained from the other by removing a token from a vertex $u$ and immediately placing it on another vertex~$v$ with the requirement that $\{u, v\}$ must be an edge of the graph. The token is then said to \emph{slide} from vertex $u$ to vertex $v$ along the edge $\{u, v\}$. Generally speaking, in the {\sc Token Sliding} problem, we are given a graph $G$ and two independent sets $I_s$ and $I_t$ of $G$. The goal is to decide whether there exists a sequence of slides (a \emph{reconfiguration sequence}) that transforms $I_s$ to $I_t$.
The problem has been extensively studied under the combinatorial reconfiguration framework~\cite{DBLP:conf/wg/BonamyB17,DBLP:conf/swat/BonsmaKW14,DBLP:conf/isaac/DemaineDFHIOOUY14,DBLP:conf/isaac/Fox-EpsteinHOU15,DBLP:conf/tamc/ItoKOSUY14,DBLP:journals/tcs/KaminskiMM12,DBLP:journals/jcss/LokshtanovMPRS18}. It is known that the problem is \PPSPACE-complete, even on restricted graph classes such as graphs of bounded bandwidth (and hence pathwidth)~\cite{DBLP:journals/jcss/Wrochna18}, planar graphs~\cite{DBLP:journals/tcs/HearnD05}, split graphs~\cite{DBLP:journals/mst/BelmonteKLMOS21}, and bipartite graphs~\cite{DBLP:journals/talg/LokshtanovM19}. However, {\sc Token Sliding} can be decided in polynomial time on trees~\cite{DBLP:conf/isaac/DemaineDFHIOOUY14}, interval graphs~\cite{DBLP:conf/wg/BonamyB17}, bipartite permutation and bipartite distance-hereditary graphs~\cite{DBLP:conf/isaac/Fox-EpsteinHOU15}, and line graphs~\cite{DBLP:journals/tcs/ItoDHPSUU11}. 

In the {\sc Token Jumping} problem, introduced by Kami\'{n}ski et al.~\cite{DBLP:journals/tcs/KaminskiMM12}, we drop the restriction that the token should move along an edge of $G$ and instead we allow it to move to any vertex of $G$ provided it does not break the independence of the set of tokens. That is, a single reconfiguration step consists of first removing a token on some vertex $u$ and then immediately adding it back on any other vertex~$v$, as long as no two tokens become adjacent. The token is said to \emph{jump} from vertex $u$ to vertex $v$.
{\sc Token Jumping} is also \PPSPACE-complete on graphs of bounded bandwidth~\cite{DBLP:journals/jcss/Wrochna18} and planar graphs~\cite{DBLP:journals/tcs/HearnD05}. Lokshtanov
and Mouawad~\cite{DBLP:journals/talg/LokshtanovM19} showed that, unlike {\sc Token Sliding}, which is \PPSPACE-complete on bipartite graphs, the {\sc Token Jumping} problem becomes \NPP-complete on bipartite graphs. 
On the positive side, it is ``easy'' to show that {\sc Token Jumping} can be decided in polynomial-time on trees (and even on split/chordal graphs) since we can simply jump tokens 
to leaves (resp. vertices that only appear in the bag of a leaf in the clique tree) to transform one independent set into another.

In this paper we focus on the parameterized complexity of the {\sc Token Sliding} problem on graphs where cycles with prescribed lengths are forbidden. 
Given an \NPP-hard problem, parameterized complexity permits to refine the notion of hardness; does the hardness come from the whole instance or from a small parameter? 
A problem $\Pi$ is \FPT\ (fixed-parameter tractable) parameterized by~$k$ if one can solve it in time $f(k) \cdot poly(n)$, for some computable function $f$. In other words, the combinatorial
explosion can be restricted to the parameter $k$. In the rest of the paper, our parameter $k$ will be the size of the independent set (i.e.\ the number of tokens). 
{\sc Token Sliding} is known to be \WO-hard parameterized by $k$ on general~\cite{DBLP:journals/jcss/LokshtanovMPRS18} and bipartite~\cite{DBLP:journals/algorithmica/BartierBDLM21} graphs. It remains \WO-hard on $\{C_4, \dots, C_p\}$-free graphs for any $p \in \mathbb{N}$~\cite{DBLP:journals/algorithmica/BartierBDLM21} and becomes \FPT\ parameterized by $k$ on bipartite $C_4$-free graphs. 
The {\sc Token Jumping} problem is \WO-Hard on general graphs~\cite{DBLP:conf/tamc/ItoKOSUY14} and is \FPT\ when parameterized by $k$ on graphs of girth five or more~\cite{DBLP:journals/algorithmica/BartierBDLM21}. For graphs of girth four, it was shown that {\sc Token Jumping} being \FPT\ would imply that {\it Gap-ETH}, an unproven computational hardness hypothesis, is false~\cite{DBLP:conf/iwpec/AgrawalAD21}.
Both {\sc Token Jumping} and {\sc Token Sliding} were recently shown to be \XL-complete~\cite{DBLP:conf/iwpec/BodlaenderGS21}. 

\paragraph*{Our result.} The complexity of the {\sc Token Jumping} problem parameterized by $k$ is settled with regard to the girth of the graph, i.e., the problem is unlikely to be \FPT\ for graphs of girth four or less and \FPT\ for graphs of girth five or more. For {\sc Token Sliding}, it was only known that the problem is \WO-hard for graphs of girth four or less and the authors in~\cite{DBLP:journals/algorithmica/BartierBDLM21} posed the question of whether there exists a constant $p$ such that the problem becomes fixed-parameter tractable on graphs of girth at least $p$. We answer their question positively and prove that the problem is indeed \FPT\ for graphs of girth five or more, which establishes a full classification of the tractability of {\sc Token Sliding} parameterized by the number of tokens based on the girth of the input graph. 

\paragraph*{Our methods.}
Our result extends and builds on the recent \emph{galactic reconfiguration} framework introduced by Bartier et al.~\cite{DBLP:journals/corr/abs-2204-05549} to show that {\sc Token Sliding} is \FPT\ on graphs of bounded degree, chordal graphs of bounded clique number, and planar graphs. Let us briefly describe the intuition behind the framework and how we adapt it for our use case. One of the main reasons why the {\sc Token Sliding} problem is believed to be ``harder'' than the {\sc Token Jumping} problem is due to what the authors in~\cite{DBLP:journals/corr/abs-2204-05549} call the \emph{bottleneck effect}. Indeed, if we consider {\sc Token Sliding} on trees, there might be a lot of empty leaves/subtrees in the tree but there might be a bottleneck in the graph that prevents any other tokens from reaching these vertices. For instance, if we consider a star with one long subdivided branch, then one cannot move any tokens from the leaves of the star to the long branch while there are at least two tokens on leaves. 
That being said, if the long branch of the star is ``long enough'' with respect to $k$ then it \emph{should} be possible to reduce parts of it; as some part would be irrelevant. In fact, this observation can be generalized to many other cases. For instance, when we have a large grid minor, then whenever a token slides into the structure it should then be able to slide freely within the structure (while avoiding conflicts with any other tokens in that structure).  
However, proving that a structure can be reduced in the context of reconfiguration is usually a daunting task due to the many moving parts. To overcome this problem, the authors in~\cite{DBLP:journals/corr/abs-2204-05549} introduce a new type of vertices called \emph{black holes}, which can simulate the behavior of a large grid minor by being able to \emph{absorb} as many tokens as they see fit; and then \emph{project} them back as needed. 

Since we need to maintain the girth property\footnote{This is not the only reason we opted to not use black holes; introducing black holes in our algorithm complicates parts of the analysis.}, we do not use the notion of black holes and instead show that when restricted to graphs of girth five or more we can efficiently find structures that behave like large grid minors (from the discussion above) and replace them with subgraphs of size bounded by a function of $k$ that can absorb/project tokens in a similar fashion (and do not decrease the girth of the graph). We note that our strategy for reducing such structures is not limited to graphs of high girth and could in principle apply to any graph. 

At a high level, our \FPT\ algorithm can then be summarized as follows. 
We let $(G, k, I_s, I_t)$ denote an instance of the problem, where $G$ has girth five or more. In a first stage, we show that we can always find a reconfiguration sequence from $I_s$ to $I_s'$ and from $I_t$ to $I_t'$ such that each vertex $v \in I_s' \cup I_t'$ has degree bounded by some function of $k$. This immediately implies that we can bound the size of $L_1 \cup L_2$, where $L_1 = I_s' \cup I_t'$ and $L_2 = N_G(I_s' \cup I_t')$. In a second stage, we show that every connected component $C$ of $L_3 = V(G) \setminus (L_1 \cup L_2)$ can be classified as either a \emph{degree-safe component}, a \emph{diameter-safe component}, a \emph{bad component}, or a \emph{bounded component}. The remainder of the proof consists in showing that degree-safe and diameter-safe components behave like large grid minors and can be replaced by bounded-size gadgets. We then show that bounded components and bad components will eventually have bounded size and we then conclude the algorithm by showing how to bound the total number of components in $L_3$. 

Finally, we note that many interesting questions remain open. In particular, it remains open whether {\sc Token Sliding} admits a (polynomial) kernel on graphs of girth five or more and whether the problem remains tractable if we forbid cycles of length $p \mod q$, for every pair of integers $p$ and $q$, or if we exclude odd cycles. 

%% file: prelim.tex
\section{Preliminaries}
We denote the set of natural numbers by $\mathbb{N}$.
For $n \in \mathbb{N}$ we let $[n] = \{1, 2, \dots, n\}$.

\paragraph*{Graphs.} We assume that each graph $G$ is finite, simple, and undirected.
We let~$V(G)$ and $E(G)$ denote the vertex set and edge set of $G$, respectively. 
The {\em open neighborhood} of a vertex $v$ is denoted by $N_G(v) = \{u \mid \{u,v\} \in E(G)\}$ and the {\em closed neighborhood} by $N_G[v] = N_G(v) \cup \{v\}$. 
For a set of vertices $Q \subseteq V(G)$, we define $N_G(Q) = \{v \not\in Q \mid \{u,v\} \in E(G), u \in Q\}$ and $N_G[Q] = N_G(Q) \cup Q$. 
The subgraph of $G$ induced by $Q$ is denoted by $G[Q]$, where $G[Q]$ has vertex set~$Q$ and edge set $\{\{u,v\} \in E(G) \mid u,v \in Q\}$. 
We let $G - Q = G[V(G) \setminus Q]$.

A {\em walk} of length $\ell$ from $v_0$ to $v_\ell$ in $G$ is a vertex sequence $v_0, \ldots, v_\ell$, such that for all $i \in \{0, \ldots, \ell-1\}$, $\{v_i,v_{i + 1}\} \in E(G)$.
It is a {\em path} if all vertices are distinct. 
It is a {\em cycle} if $\ell \geq 3$, $v_0 = v_\ell$, and $v_0, \ldots, v_{\ell - 1}$ is a path.
A path from vertex $u$ to vertex $v$ is also called a {\em $uv$-path}.
For a pair of vertices~$u$ and~$v$ in $V(G)$, by $\textsf{dist}_G(u,v)$ we denote the {\em distance} or length of a shortest $uv$-path in $G$ (measured in number of edges and set to $\infty$ if $u$ and $v$ belong to different connected components).
The {\em eccentricity} of a vertex $v \in V(G)$, $\textsf{ecc}(v)$, is equal to $\max_{u \in V(G)}(\textsf{dist}_G(u,v))$.
The {\em diameter} of $G$, $\textsf{diam}(G)$, is equal to $\max_{v \in V(G)}(\textsf{ecc}(v))$. 
The \emph{girth} of $G$, $\textsf{girth}(G)$, is the length of a shortest cycle contained in $G$. If the graph does not contain any cycles (that is, it is a forest), its girth is defined to be infinity.

\paragraph*{Reconfiguration.} In the {\sc Token Sliding} problem we are given a graph $G = (V,E)$ and two independent sets $I_s$ and $I_t$ of $G$, each of size $k \geq 1$. The goal is to determine whether there exists a sequence $\langle I_0, I_1, \ldots, I_\ell \rangle$ of independent sets of size $k$ such that $I_s = I_0$, $I_\ell = I_t$, and $I_i \Delta I_{i+1} = \{u, v\} \in E(G)$ for all $i\in \{0,\ldots, \ell-1\}$. In other words, if we view each independent set as a collection of tokens placed on a subset of the vertices of $G$, then the problem asks for a sequence of independent sets which transforms $I_s$ to $I_t$ by individual token slides along edges of $G$ which maintain the independence of the sets. Note that {\sc Token Sliding} can be expressed in terms of a \emph{reconfiguration graph} $\mathcal{R}(G,k)$. $\mathcal{R}(G,k)$ contains a node for each independent set of $G$ of size exactly $k$. We add an edge between two nodes whenever the independent set corresponding to one node can be obtained from the other by a single reconfiguration step. That is, a single token slide corresponds to an edge in $\mathcal{R}(G,k)$. The {\sc Token Sliding} problem asks whether $I_s, I_t \in V(\mathcal{R}(G,k))$ belong to the same connected component of~$\mathcal{R}(G,k)$. 

%% file: reducing.tex
\section{Reducing the graph}
Let $(G, k, I_s, I_t)$ be an instance of {\sc Token Sliding}, where $G$ has girth five or more. The aim of this section is to bound the size of the graph by a function of $k$. We start with a very simple reduction rule that allows us to get rid of most twin vertices in the graph. Two vertices $u,v \in V(G)$ are said to be \emph{twins} if $u$ and $v$ have the same set of neighbours, that is, if $N(u) = N(v)$.

\begin{lemma}\label{lem:removal_of_twin_nodes}
Assume $u, v \in V(G) \setminus (I_s \cup I_t)$ and $N(u) = N(v)$. 
Then $(G, k, I_s, I_t)$ is a yes-instance if and only if $(G - \{v\}, k, I_s, I_t)$ is a yes-instance. 
\end{lemma}

\begin{proof}
Since  $u, v \in V(G) \setminus (I_s \cup I_t)$ and $G - \{v\}$ is an induced subgraph of $G$, it follows that if there exists a reconfiguration sequence $\mathcal{S} = \langle I_0,I_1,\ldots,I_{\ell-1},I_\ell \rangle$ from $I_s$ to $I_t$ in $G - \{v\}$, then the same sequence remains valid in $G$. 

Now assume that there exists a sequence $\mathcal{S} = \langle I_0,I_1,\ldots,I_{\ell-1},I_\ell \rangle$ from~$I_s$ to~$I_t$ in $G$. Since $u, v \in V(G) \setminus (I_s \cup I_t)$, in $I_s$ there are no tokens on $u$ and $v$ and the same holds for $I_t$. Hence, if there exists $I_i$, $1 \leq i \leq \ell-1$ such that $v \in I_i$, then $u \not\in I_i$. The reason is that a token can
be moved to $u$ only via $N(u)$. By assumption $N(u) = N(v)$ and 
$N(v)$ is blocked by the token on $v$. This implies that we can always choose to slide the token to $u$ instead of $v$, as needed. 
\qed
\end{proof}

Note that in a graph of girth at least five twins can have
degree at most one. 

Given~\cref{lem:removal_of_twin_nodes}, we assume in what follows that twins have been reduced. In other words, we let 
$(G, k, I_s, I_t)$ be an instance of {\sc Token Sliding} where $G$ has girth five or more and twins not in $I_s \cup I_t$ have been removed. We now partition our graph into three sets $L_1 = I_s \cup I_t$, $L_2 = N_G(L_1)$, and $L_3 = V(G) \setminus (L_1 \cup L_2)$. 

\begin{lemma}\label{lem:bound_l2_l3_degree}
If $u \in L_2 \cup L_3$, then $u$ has at most $|L_1| \leq 2k$ neighbors in $L_1 \cup L_2$, i.e., $|N_{L_1 \cup L_2}(u)| \leq 2k$.
\end{lemma}

\begin{proof}
Assume $u_1$ is a vertex in $L_2$ and $u_2 \in N_{L_2}(u_1)$ is a neighbor of $u_1$ in~$L_2$. If~$u_1$ and $u_2$ have a common neighbor $u_3 \in L_1$, then this would imply the existence of a triangle in $G$, a contradiction. 

Now assume $u_1 \in L_3$ and assume $u_2, u_3 \in N_{L_2}(v_1)$ are two neighbors of $u_1$ in~$L_2$. If $u_2$ and $u_3$ have a common neighbor $u_4 \in L_1$ this would imply the existence of a $C_4$ in $G$, a contradiction. 

Hence, for any vertex $u \in L_2 \cup L_3$ we have $N_{L_1}(v) \cap N_{L_1}(w) = \emptyset$ for all $v, w \in N_{L_2}[u]$. 
Since each vertex in $L_2$ has at least one neighbor in $L_1$ by definition, each vertex $u \in L_2 \cup L_3$ can have at most one neighbor in $L_2$ for each of its non-neighbor in $L_1$, for a total of $|L_1| \leq 2k$ neighbors in $L_1 \cup L_2$.
\qed
\end{proof}

\subsection{Safe, bounded, and bad components}
Given $G$ and the partition $L_1 = I_s \cup I_t$, $L_2 = N_G(L_1)$, and $L_3 = V(G) \setminus (L_1 \cup L_2)$ we now classify components of $G[L_3]$ into four different types. 

\begin{definition}\label{l3_component_types}
Let $C$ be a maximal connected component in $G[L_3]$. 

\vspace{-2mm}
\begin{itemize}
\item We call $C$ a \emph{diameter-safe component} whenever $\textsf{diam}(G[V(C)]) > k^3$.\\[-2mm]
\item We call $C$ a \emph{degree-safe component} whenever $G[V(C)]$ has a vertex $u$ with at least $k^2 + 1$ neighbors $X$ in $C$ and at least $k^2$ vertices of $X$ have degree two in $G[V(C)]$.\\[-2mm]
\item We call $C$ a \emph{bounded component} whenever $\textsf{diam}(G[V(C)]) \leq k^3$ and no vertex of $C$ has degree more than $k^2$ in $G[V(C)]$.\\[-2mm]
\item We call $C$ a \emph{bad component} otherwise.
\end{itemize}
\end{definition}

Note that every component of $G[L_3] = G - (L_1 \cup L_2)$ is  \emph{safe} (degree- or diameter-safe), bad, or bounded. 

\begin{lemma}\label{bounded_components_size}
A bounded component $C$ in $G[L_3]$ contains at most $k^{2k^3}$ vertices, i.e., $|V(C)| \leq k^{2k^3}$.
\end{lemma}

\begin{proof}
Let $T$ be a spanning tree of $C$ and let $u \in V(C)$ denote the root of $T$. Each vertex in $T$ has at most $k^2$ children given the degree bound of $C$ and the height of the tree is at most $k^3$ given the diameter bound of $C$. Hence the total number of vertices in $C$ is at most $k^{2k^3}$.
\qed
\end{proof}

We now describe a crucial property of degree-safe and diameter-safe components, which we call the \emph{absorption-projection property}. We note that this notion is similar to the notion of black holes introduced in~\cite{DBLP:journals/corr/abs-2204-05549}. The key (informal) insight is that for a safe component $C$ we can show the following:

\begin{enumerate}
    \item If there exists a reconfiguration sequence $\mathcal{S} = \langle I_0,I_1,\ldots,I_{\ell-1},I_\ell \rangle$ from $I_s$
    to~$I_t$, then we may assume that $I_i \cap N_G(V(C)) \leq 1$, for $0 \leq i \leq \ell$. 
    \item A safe component can \emph{absorb} all $k$ tokens, i.e, a safe component contains an independent set of size at least $k$ and whenever a token reaches $N_G(V(C))$ then we can (but do not have to) absorb it into $C$ (regardless of how many tokens are already in $C$). Moreover, a safe component can then \emph{project} the tokens back into its neighborhood as needed.  
\end{enumerate}

Let us start by proving the absorption-projection property for degree-safe components. An \emph{$s$-star} is a vertex with $s$ pairwise non-adjacent neighbors, which are called the leaves of the
$s$-star. A \emph{subdivided $s$-star} is an $s$-star where each edge is subdivided (replaced by a new vertex of degree two adjacent to the endpoints of the edge) any number of times. We say that each leaf of a subdivided star belongs to a \emph{branch} of the star. 

\begin{lemma}\label{subdivided_star_degree_safe_components}
Let $C$ be a degree-safe component in $G[L_3]$. Then $C$ contains an induced subdivided $k$-star where all $k$ branches have length more than one.
\end{lemma}

\begin{proof}
Since $C$ is a degree-safe component, it must contain a vertex $u$ with at least $k^2$ neighbors in $C$ and each one of these
neighbors must have another neighbor in $C$. Note that all
of these vertices must be distinct, as otherwise we could find
a cycle of length three or four. 

Let us call the distance-one and distance-two neighbors of $u$ in $C$ the first level and second level. That is, we let $N_1(u) = N_C(u) \setminus \{u\}$ and $N_2(u) = N_C(N_1(u)) \setminus (N_1(u) \cup \{u\})$.
    
Note that the first level, $N_1(u)$, is an independent set, since otherwise that would imply the existence of a triangle. Also, vertices in the second level, $N_2(u)$, cannot be connected to more than one vertex of the first level, since that would imply the existence of a $C_4$. 
    
As for the second level, it contains at least $k^2$ vertices and we can have edges between those vertices. We claim that $G_2 = G[N_2(u)]$
contains an independent set of size $k$. Assume first that 
$G_2$ contains a vertex $v$ of degree $k$. Then, since $G_2$ is
triangle free, the $k$ neighbors of $v$ form the required independent set. Otherwise, all vertices of $G_2$ have degree
at most $k-1$. We iteratively add one vertex $v$ to the independent set and remove $N[v]$ from $G_2$. This can be repeated for $k$ times leading to the required independent set. 
Therefore, we get an induced subdivided star with at least $k$ branches of length at least two and there is no edge between the different branches.
\qed
\end{proof}

\begin{figure}
\centering
\includegraphics[scale=0.4]{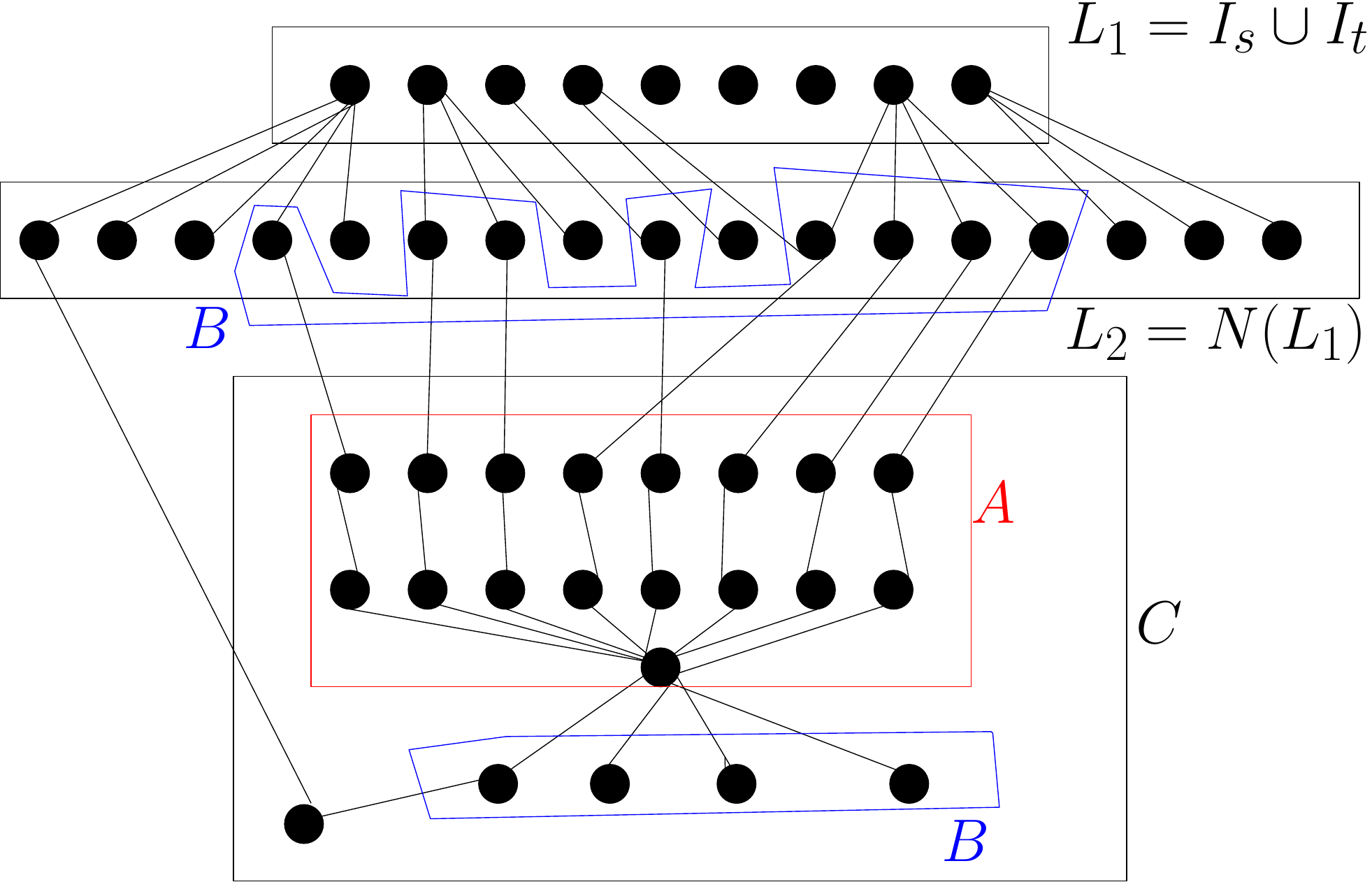}
\caption{An illustration of a degree-safe component $C$.}
\label{fig:degree_safe}
\end{figure}

\begin{lemma}\label{lem:deg-safe-behavior}
Let $C$ be a degree-safe component in $G[L_3]$ and let $A$ be an induced subdivided $k$-star contained in $C$ where all branches have length exactly two. Let $B = N_G(A)$. If $(G, k, I_s, I_t)$ is a yes-instance, then there exists a reconfiguration sequence from $I_s$ to $I_t$ in $G$ where we have at most one token on a vertex of $B$ at all times.
\end{lemma}

\begin{proof}
First, note that the existence of $A$ follows from \cref{subdivided_star_degree_safe_components} and that it is indeed the case that $I_s \cap B = I_t \cap B = \emptyset$. Let $r$ denote the root of the induced subdivided $k$-star and let $N_1$ and $N_2$ denote the first and second levels of subdivided the star, respectively. Let us explain how we can adapt a transformation $\mathcal{S}$ from $I_s$ to $I_t$ into a transformation containing at most one token on a vertex of $B$ at all times and such that, at any step, the number of tokens in $A \cup B$ in both transformations is the same and the positions of the tokens in $V(G) \setminus (A \cup B)$ are the same.

Assume that, in the transformation $\mathcal{S}$, a token is about to reach a vertex $b \in B$, that is, we consider the step right before a token is about to slide into $B$. We first move all tokens residing in $A$, if any, to the second level of their branches, i.e, to $N_2$. This is possible as $A$ is an induced subdivided star 
and there are no other tokens on $B$. Note that we can assume that there is no token on $r$ (and hence every token is on a branch and ``the branch'' of a token is well defined) since we can otherwise slide this token to one of the empty branches while $B$ is still empty of tokens. Then we proceed as follows:

\begin{itemize}
    \item If $b$ is a neighbor of the root $r$ of the subdivided star, then $b$ is not a neighbor of any vertex at the second level of $A$, since otherwise this would create a cycle of length four. Hence, we can slide the token into $b$ and then $r$ and then some empty branch of $A$ (which is possible since we have $k$ branches in $A$). 
    \item Otherwise, if $b$ has no neighbors in the first level $N_1$ of $A$, we choose a branch that has a neighbor $a$ of $b$ in $N_2$ (which exists since $b$ is not adjacent to $r$ nor $N_1$). Then, if the branch of $a$ already contains a token, we can safely slide the token into another branch by going to the first level, then the root $r$, then to another empty branch of $A$. Now we slide all tokens in $A$ to the first level of their branch and finally we slide the initial token to $b$ and then to $a$. 
    \item Finally, if $b$ has neighbors in the first level of $A$, note that it cannot have more than one neighbor in $N_1$ since that would imply the existence of a cycle of length four. Let $a$ denote the unique neighbor of $b$ in $N_1$. If the branch of~$a$ has a token on it, then we safely slide it into another empty branch. Now we slide all tokens in $A$ to the first level of their branch and finally we slide the initial token to $b$ and then to $a$. 
\end{itemize}
   
Note that all of above slides are reversible and we can therefore use a similar strategy to project tokens from $A$ to $B$. If, in $\mathcal{S}$, a token is about to leave the vertex $b \in B$, then we can similarly move a token from $A$ to $b$ and then perform the same move. Finally, if a reconfiguration step in $\mathcal{S}$ consists of moving tokens in $A\cup B$ to $A \cup B$, we ignore that step. And, if it consists of moving a token from $V(G) \setminus (A \cup B)$ to $V(G) \setminus (A \cup B)$ we perform the same step.
   
It follows from the previous procedure that whenever $(G, k, I_s, I_t)$ is a yes-instance we can find a reconfiguration sequence from $I_s$ to $I_t$ in $G$ where we have at most one token in $B$ at all times, as claimed (see~\cref{fig:degree_safe}). 
\qed
\end{proof}

\begin{corollary}\label{cor:deg_safe}
Let $C$ be a degree-safe component. If $(G, k, I_s, I_t)$ is a yes-instance, then there exists a reconfiguration sequence from $I_s$ to $I_t$ in $G$ where we have at most one token in $N(C) \subseteq L_2$ at all times. 
\end{corollary}

\begin{proof}
Assume a token slides to a vertex $c \in N(C)$ (for the first time). If $c \in B$, then the result follows from~\cref{lem:deg-safe-behavior}. Otherwise, we can follow a path $P$ contained in $C$ that leads to the root of the induced $k$-subdivided star (such a path exists since $c \in N(C)$ and $C$ is connected) and right before we reach $B$ we then again can apply~\cref{lem:deg-safe-behavior}. Note that, regardless of whether $c$ is in $B$ or not, once the token reaches $N(C)$ we can assume that it is immediately absorbed by the degree-safe component (and later projected as needed). This implies that we can always find a path $P$ to slide along such that $N[P]$ contains no tokens.  
\qed
\end{proof}

We now turn our attention to diameter-safe components and show that they have a similar absorption-projection behavior as degree-safe components. Given a component $C$ we say that a path $A$ in $C$ is a \emph{diameter path} if $A$ is a longest shortest path in $C$. 

\begin{figure}
\centering
\includegraphics[scale=0.4]{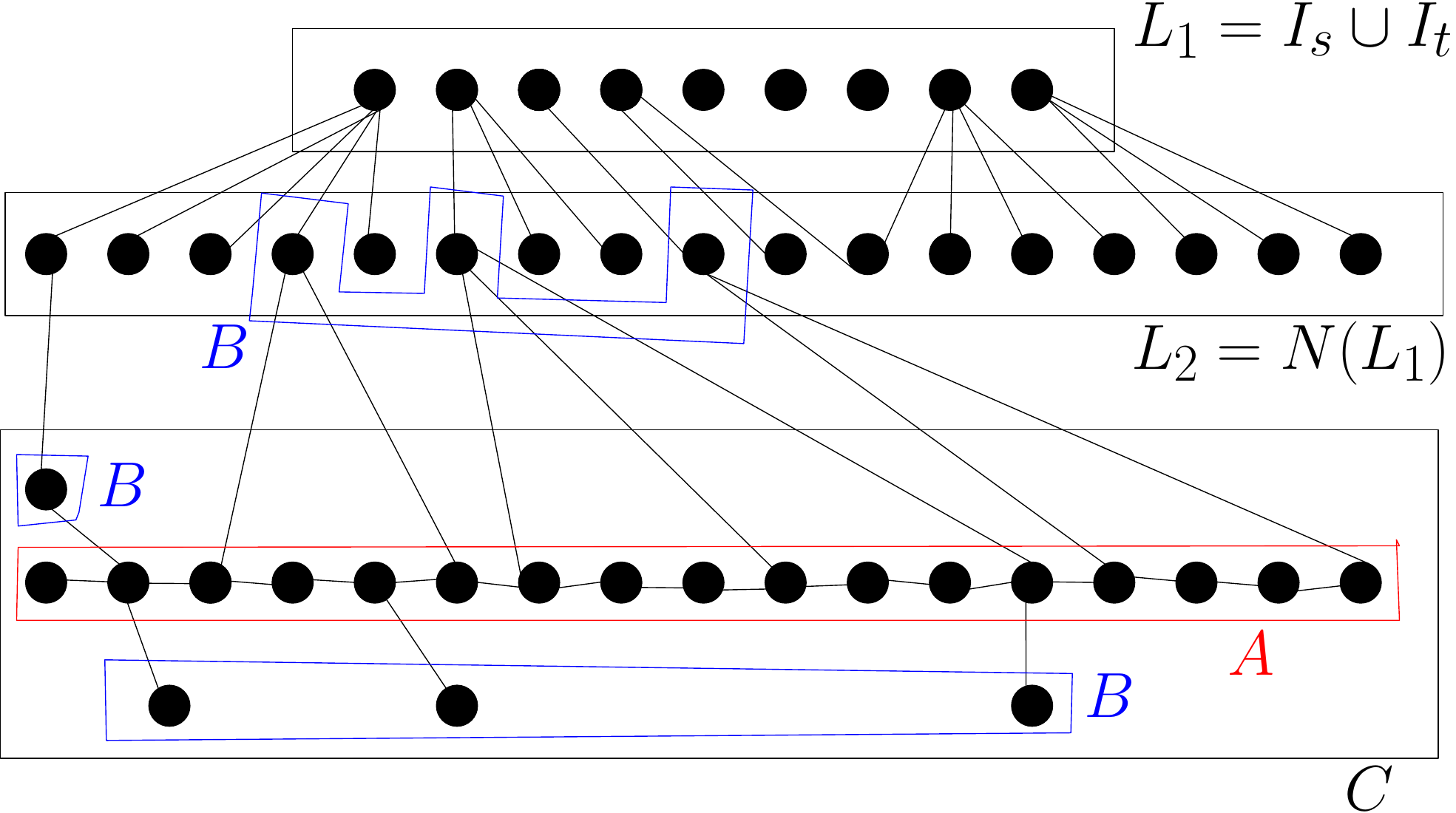}
\caption{An illustration of a diameter-safe component $C$.}
\label{fig:diam_safe}
\end{figure}

\begin{lemma}\label{lem:diam_safe_behavior}
Let $C$ be a diameter-safe component, let $A$ be a diameter path of $C$, and let $B = N_G(V(A))$. If $(G, k, I_s, I_t)$ is a yes-instance, then there exists a reconfiguration sequence from $I_s$ to $I_t$ in $G$ where we have at most one token on vertices of $B$ at all times.
\end{lemma}

\begin{proof}
As in the proof of Lemma~\ref{lem:deg-safe-behavior}, the goal will consist in proving that we can adapt a transformation $\mathcal{S}$ from $I_s$ to $I_t$ into a transformation containing at most one token on a vertex of $B$ at all times and such that, at any step, the number of tokens in $A \cup B$ in both transformations is the same and the positions of the tokens in $V(G) \setminus (A \cup B)$ are the same. As in the proof of Lemma~\ref{lem:deg-safe-behavior}, all the tokens in $A \cup B$ will be absorbed into $A$ (and later projected back as needed) and it suffices to explain how we can move the tokens on $A$ when a new token wants to enter in $B$ or leave into $B$.

We know that two non-consecutive vertices in $A$ cannot be adjacent by minimality of the path. Now assume a token $t$ is about to reach a vertex $b \in B$. Note that neighbors of $b$ in $A$ are pairwise at distance at least three in $A$, since otherwise that would create a cycle of length less than five. We call the intervals between consecutive neighbor of $b$ \emph{gap intervals} (with respect to $b$). 
    
If $b$ has more than $k$ neighbors in $A$, then we can put the already in $A$ tokens (at most $k - 1$ of them) in the at most $k-1$ first gap intervals. Indeed, since there is no token on $B$ and $A$ is an induced path, we can freely move tokens where we want. Then we can slide the token $t$ to $b$, since none of its neighbors in $A$ have a token on them, and then slide it to the next neighbor of $b$ in $A$ since it has more than $k$ neighbors.
    
Otherwise, $b$ has at most $k$ neighbors in $A$. Hence there are at most $k + 1$ gap intervals in $A$ (with respect to $b$). The average number of vertices in the gap intervals (assuming $k\geq 4$) is 
    
$$\alpha = \frac{\textsf{diam}(C) - |N_A(b)|}{|N_A(b)| + 1} \geq \frac{k^3 - k}{k + 1} \geq 2k.$$

Hence at least one gap interval has length at least $\alpha$ and therefore we can slide all tokens currently in $A$ (at most $k-1$ of them) into this gap interval in such a way no token is on the border of the gap interval (since the gap interval contains an independent set of size at least $k-1$ which does not contain an endpoint of the gap interval).
Now we can simply slide the token $t$ onto $b$ and then onto any of the neighbors of $b$ in $A$.
    
Combined with the fact that the above strategy can also be applied to project a token from $A$ to $B$, it then follows that whenever $(G, k, I_s, I_t)$ is a yes-instance we can find a reconfiguration sequence from $I_s$ to $I_t$ in $G$ where we have at most one token in $B$ at all times, as claimed (see~\cref{fig:diam_safe}). 
\qed
\end{proof}

\begin{corollary}\label{cor:diam_safe}
Let $C$ be a diameter-safe component. If $(G, k, I_s, I_t)$ is a yes-instance then there exists a reconfiguration sequence from $I_s$ to $I_t$ where we have at most one token in $N(C) \subseteq L_2$ at all times.
\end{corollary}

\begin{proof}
We follow the same strategy as for the degree-safe components. When a token reaches a vertex in $N(C)$ (for the first time), if it belongs to $B$ the result follows from~\cref{lem:diam_safe_behavior}. Otherwise we can move along a path in $C$ to the closest vertex of the diameter path to reach $B$ and then the result again follows from~\cref{lem:diam_safe_behavior}.
\qed
\end{proof}

Putting~\cref{cor:deg_safe} and~\cref{cor:diam_safe} together, we know that if $(G, k, I_s, I_t)$ is a yes-instance, then there exists a reconfiguration sequence from $I_s$ to $I_t$ where we have at most one token in $N(C) \subseteq L_2$ at all times, where $C$ is either a degree-safe or a diameter-safe component. We now show how to reduce a safe component $C$ by replacing it by another smaller subgraph that we denote by $H$. 

\begin{figure}
\centering
\includegraphics[scale=0.4]{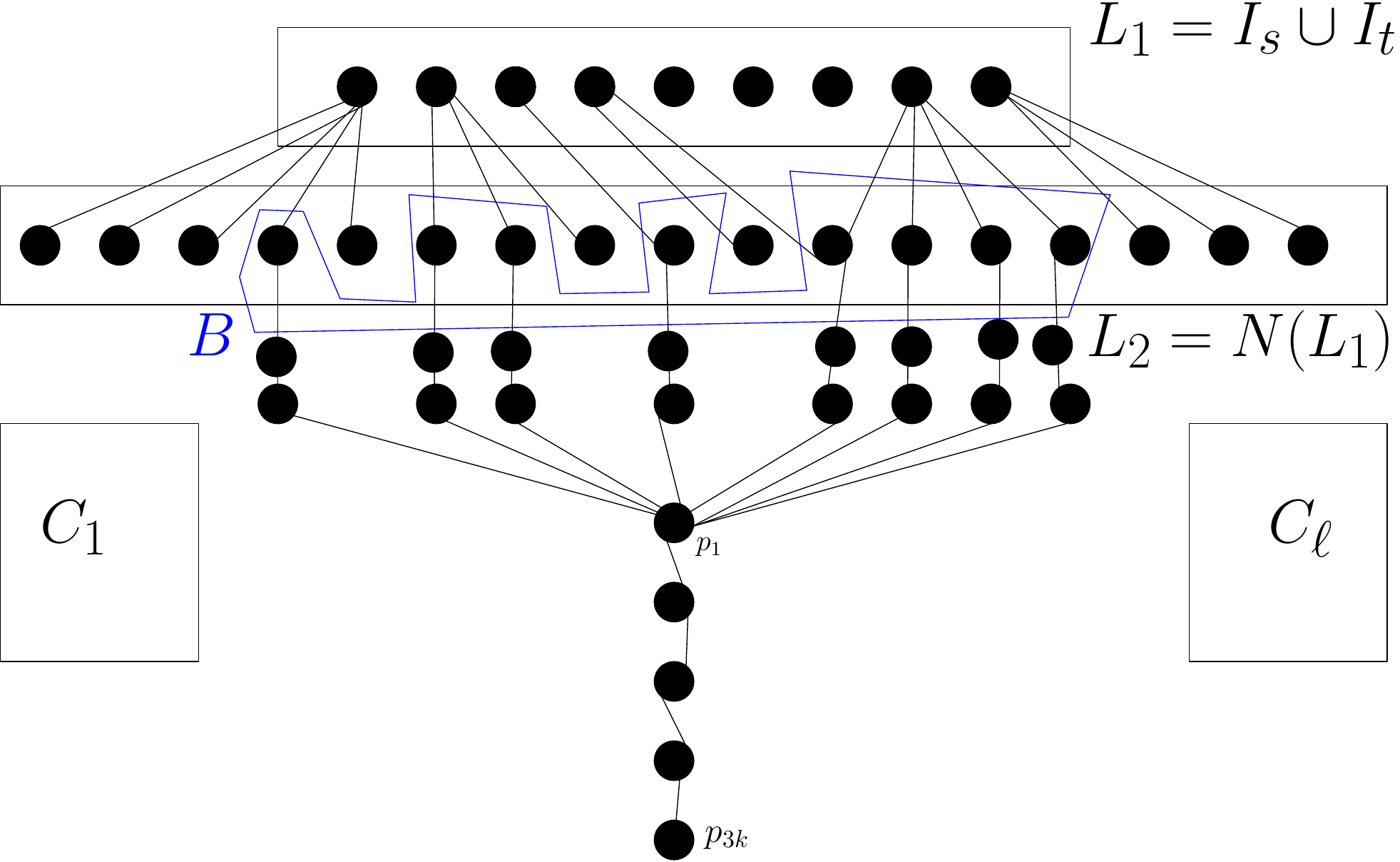}
\caption{An illustration of the replacement gadget for a safe component $C$.}
\label{fig:replace}
\end{figure}

\begin{lemma}\label{bounding_safe_components}
Let $C$ be a safe component in $G[L_3]$ and let $G'$ be the graph obtained from $G$ as follows:
\begin{itemize}
    \item Delete all vertices of $C$ (and their incident edges).
    \item For each vertex $v \in N(C) \subseteq L_2$ add two new vertices $v'$ and $v''$ and add the edges $\{v, v'\}$ and $\{v', v''\}$. 
    \item Add a path of length $3k$ consisting of new vertices $p_1$ to $p_{3k}$. 
    \item Add an edge $\{p_1, v''\}$ for every vertex $v''$.
\end{itemize}
Note that this new component has size $3k + |2N(C)|$ (see~\cref{fig:replace}). We claim that 
$(G, k, I_s, I_t)$ is a yes-instance if and only if $(G', k, I_s, I_t)$ is a yes-instance. 
\end{lemma}

\begin{proof}
First, we note that replacing $C$ with this new component, $H$, cannot create cycles of length less than five. This follows from the fact that all the vertices at distance one or two from $p_1$ have distinct neighbors.
    
Assume $(G, k, I_s, I_t)$ is a yes-instance. Then, by~\cref{cor:deg_safe} and~\cref{cor:diam_safe}, we know that there exists a reconfiguration sequence from $I_s$ to $I_t$ in $G$ where we have at most one token in $N(C) \subseteq L_2$ at all times, where $C$ is either a degree-safe or a diameter-safe component. Hence, we can mimic the reconfiguration sequence from $I_s$ to $I_t$ in $G'$ by simply projecting tokens onto the path of length $3k$ in each of the safe components that we replaced. 

Now assume that $(G', k, I_s, I_t)$ is a yes-instance. By the same arguments, and combined with the fact that a safe component $C$ can absorb/project the same number of tokens as its replacement component $H$, we can again mimic the reconfiguration sequence of $G'$ in $G$. 
\qed
\end{proof}

\subsection{Bounding the size of $L_2$}
Having classified the components in $L_3$ and the edges between $L_2$ and $L_3$, our next goal is to bound the size of $L_2$, which until now could be arbitrarily large. 
We know that vertices in $L_2$ are the neighbors of vertices in $L_1$, hence the size of $L_2$ will grow whenever there are vertices in $L_1$ with arbitrarily large degrees. Bounding $L_2$ will therefore be done by first proving the following lemma.

\begin{lemma}\label{bounding_l1_degrees}
Assume a vertex $u$ in $L_1 = I_s \cup I_t$ has degree greater than $2k^2$. Moreover, assume, without loss of generality, that $u \in I_s$. Then, there exists $I_s'$ such that $I_s \triangle I_s' = \{u,u'\}$, $u'$ has degree at most $2k^2$, and the token on $u$ can slide to $u'$. 
\end{lemma}

\begin{proof}
First note that from such a vertex $u \in I_s$ we can always slide to a vertex in $L_2$.
Indeed, for every $v$, $|N(u) \cap N(v)| \le 1$ by the assumption on the girth of the graph. Thus, since the degree of $u$ is larger than the number of tokens, there exists at least one vertex in $L_2$ that the token on $u$ can slide to. 
    
If we slide to a vertex $v \in L_2$ of degree at most $2k^2$, then we are done (we set~$u' = v$). 
Otherwise, by~\cref{lem:bound_l2_l3_degree}, we know that most of the neighbors of $v$ are in~$L_3$; since $v$ has degree greater than $2k^2$ and at most $2k$ of its neighbors are in $L_1 \cup L_2$. Hence, we are guaranteed at least one neighbor $w$ of $v$ in some component of $L_3$.
    
If we reach a bounded component $C$, i.e., if $w$ belongs to a bounded component, then all vertices of $C$ (including $w$) have at most $k^2$ neighbors in $C$ and have at most $2k$ neighbors in $L_2$ (by~\cref{lem:bound_l2_l3_degree}) and thus we can set $u' = w$. 

If we reach a bad component $C$, then we know that $C$ has a vertex $b$ with at least $k^2 + 1$ neighbors in $C$ and at most $k^2 - 1$ of those neighbors have other neighbors in $C$. 
Let $z$ denote a vertex in the neighborhood of $b$ that does not have other neighbors in $C$. By~\cref{lem:bound_l2_l3_degree}, $z$ will have degree at most $2k + 1$ and we can therefore let $u' = z$.

Finally, if we reach a safe component, then after our replacement such components contain a lot of vertices of degree exactly two and we can therefore slide to any such vertex, which completes the proof.
\qed
\end{proof}

After exhaustively applying~\cref{bounding_l1_degrees}, each time relabeling vertices in $L_1$, $L_2$ and $L_3$ and replacing safe components as described in~\cref{bounding_safe_components}, we get an equivalent instance where the maximum degree in $L_1$ is at most $2k^2$ and hence we get a bound on the size of $L_2$. We conclude this section with the following lemma.

\begin{lemma}\label{l2_bound}
Let $(G, k, I_s, I_t)$ be an instance of {\sc Token Sliding}, where $G$ has girth at least five. Then we can compute an equivalent instance $(G', k, I_s', I_t')$, where $G'$ has girth at least five, $|L_1 \cup L_2| \leq 2k + 4k^3 = O(k^3)$, and each safe component of $G$ is replaced in $G'$ by a component with at most $3k + 8k^3 = O(k^3)$ vertices. 
\end{lemma}

\subsection{Bounding the size of $L_3$}

We have proved that the number of vertices in $L_1$ and $L_2$ is bounded by a function of $k$, namely $|L_1 \cup L_2| = O(k^3)$. We have also shown that every safe or bounded component in $L_3$ has a bounded number of vertices, namely safe components have $O(k^3)$ vertices and bounded components have at most $k^{2k^3}$ vertices. We still need to show that $L_3$ is bounded. We start by showing that bad components become bounded after bounding $L_2$:

\begin{lemma}\label{lem:bound_bad}
Let $(G, k, I_s, I_t)$ be an instance where $G$ has girth at least five, $|L_1 \cup L_2| \leq 2k + 4k^3 = O(k^3)$, and each safe component has at most $3k + 8k^3 = O(k^3)$ vertices. Then, every bad component in that instance has at most $k^{O(k^3)}$ vertices.
\end{lemma}

\begin{proof}
Let $C$ be a bad component, hence $\textsf{diam}(C) \leq k^3$ since $C$ is not diameter-safe. Let $v \in V(C)$ be a vertex in $C$ whose degree is $d > k^2$. Since $C$ is not a degree-safe component $v$ can have at most $k^2 - 1$ neighbors in $C$ that have other neighbors in $C$. Hence, at least $d - (k^2 - 1) = d - k^2 + 1$ neighbors of $v$ will have only $v$ as a neighbor in $C$ and all their other neighbors must be in $L_2$. Since, by~\cref{lem:removal_of_twin_nodes}, we can assume that $L_3$ contains no twin vertices, $d - k^2$ of the neighbors of $v$ in $C$ must have at least one neighbor in $L_2$. 
But we know that~$L_2$ has size $O(k^3)$ and if two neighbors of $v$ had a common neighbor in $L_2$, this would imply the existence of a cycle of length four. Therefore, $d$ must be at most~$O(k^3)$. Having bounded the degree and diameter of bad components, we can now apply the same argument as in the proof of~\cref{bounded_components_size}.
\qed
\end{proof}

Since bounded and bad components now have the same asymptotic number of vertices, in what follows we refer to both of them as bounded components. What remains to show is that the number of safe and bounded components is also bounded by a function of $k$ and hence $L_3$ and the whole graph will have size bounded by a function of $k$.

\begin{definition}\label{equivalent_components}
Let $C_1$ and $C_2$ be two components in $G[L_3]$ and $B_1$ and $B_2$ be their respective neighborhoods in $L_2$. We say $C_1$ and $C_2$ are equivalent whenever $B_1 = B_2 = B$ and $G[V(C_1) \cup B]$ is isomorphic to $G[V(C_2) \cup B]$ by an isomorphism that fixes $B$ point-wise. We let $\beta(G)$ denote the number of equivalence classes of bounded components  and we let $\sigma(G)$ denote the number of equivalence classes of safe components.
\end{definition}

We are now ready to prove a crucial result for bounding $L_3$.

\begin{lemma}\label{redundant_component_removal}
Let $S_1$ and $S_2$ be equivalent safe components and let $B_1$, $\ldots$, $B_{k + 1}$ be equivalent bounded components. Then, $(G, k, I_s, I_t)$, $(G - V(S_2), k, I_s, I_t)$ and $(G - V(B_{k + 1}), k, I_s, I_t)$ are equivalent instances.
\end{lemma}

\begin{proof}
Removing vertices from the graph preserves no-instances. As for yes-instances, we will prove equivalence for safe and bounded components separately.

Assume a token reaches the neighborhood of $S_1$ and $S_2$ (they have the same neighborhood). Whether the token slides to either of them is irrelevant because both can hold all the tokens together and have the same behavior regarding entering from $L_2$ and leaving to $L_2$. Hence, from~\cref{cor:deg_safe} and~\cref{cor:diam_safe}, we can always choose to slide to $S_1$ and never to $S_2$ and therefore removing $S_2$ will preserve yes-instances.
    
Assume a token reaches the neighborhood of all $B_i$'s (they have the same neighborhood). The components not being empty implies that each one can hold at least one token if it can, and hence we can always choose to slide the tokens to one of the first $k$ components since it will be enough to hold all tokens. Therefore removing $B_{k+1}$ will preserve yes-instances.
\qed
\end{proof}

After exhaustively removing equivalent components
as described in~\cref{redundant_component_removal} we obtain
the following corollary. 

\begin{corollary}\label{number_of_components}
There are at most $k\beta(G)$ bounded components and $\sigma(G)$ safe components. 
\end{corollary}

This leads to the final lemma. 

\begin{lemma}\label{equivalency_classes_bound}
We have $\beta(G) = 2^{k^{O(k^3)}}$, $\sigma(G) = 2^{O(k^6)}$,  
$|L_3| \leq k^{O(k^3)}2^{k^{O(k^3)}} + k^3 2^{O(k^6)} = 2^{k^{O(k^3)}}$, and  
$|V(G)| = |L_1| + |L_2| + |L_3| = 2^{k^{O(k^3)}}$.
\end{lemma}

\begin{proof}
Since $L_2$ and safe components have $O(k^3)$ size (from~\cref{l2_bound}) then safe components along with their neighbors in $L_2$ have size $O(k^3)$. Hence there are 
$2^{O(k^6)}$ equivalence classes of safe components.
    
Since bounded components have size $k^{O(k^3)}$  (from~\cref{bounded_components_size}) the bounded components along with their neighbors in $L_2$ have size $k^{O(k^3)}$ and hence there are $2^{k^{O(k^3)}}$ equivalence classes of bounded components.
    
Finally, using the fact that there are $2^{n \choose 2}$ graphs with $n$ vertices combined with~\cref{number_of_components}, we get the desired bound on $L_3$, which implies the desired bound on the size of $V(G)$.
\qed
\end{proof}

%% file: algorithm.tex
\section{The algorithm}

\subsection{Outline}
Now that we have bounded the size of $G$ by $f(k) = 2^{k^{O(k^3)}}$ we describe below the complete algorithm for solving an instance $(G, k, I_s, I_t)$ of the {\sc Token Sliding} problem, where $G$ has girth five or more. 
\begin{enumerate}
    \item Bound the graph size;
    \begin{enumerate}
        \item Remove twin vertices as described in~\cref{lem:removal_of_twin_nodes}; 
        \item Repeat the following while $L_1$ has a vertex of degree greater than $2k^2$ or there exists an unbounded safe component in $L_3$:
        \begin{itemize}
            \item Find safe components as described in~\cref{l3_component_types};
            \item Replace safe components as described in~\cref{bounding_safe_components};
            \item Find a vertex $u \in L_1$ with degree greater than $2k^2$;
            \item Slide the token to a vertex of degree at most $2k^2$ (\cref{bounding_l1_degrees});
        \end{itemize}
        \item Test all pairs of $L_3$ components for equivalence (\cref{equivalent_components});
        \item Partition the components into equivalence classes;
        \begin{itemize}
            \item For classes containing a safe component, keep one component and remove the others from the graph (\cref{redundant_component_removal});
            \item For each other class, keep $k$ components and remove the others from the graph. If there are already less than $k$ components then do nothing (\cref{redundant_component_removal});
        \end{itemize}
    \end{enumerate}
    
    \item Build the graph $\mathcal{R}(G, k)$;
    \begin{itemize}
        \item $\mathcal{R}(G, k)$ will have a node for each independent set of $G$ of size $k$;
        \item Two nodes $I, J \in \mathcal{R}(G, k)$ will be connected by an edge if the corresponding independent sets are adjacent with respect to the token slide definition, namely $I \Delta J = \{u, v\} \in E(G)$;
    \end{itemize}
    
    \item Run a breadth-first search (BFS) traversal on $\mathcal{R}(G, k)$ with source $I_s$ and destination $I_t$. Return {\it true} if the two are in the same component and {\it false} otherwise;
\end{enumerate}

\subsection{Analysis}

\paragraph*{Complexity of step (1).}
Step (a), removing twin vertices, can be naively implemented to run in $O(n^3)$-time.
Going to step (b), finding degree-safe components will take $O(n)$-time by simply checking the degrees of all vertices in a component. As for diameter-safe components, we can find them in $O(n^2)$-time by finding for each vertex $u$ in a component $C$ the vertex $v$ furthest away from $u$ in $C$ using a~BFS. Replacing a component can be done in $O(n)$-time. Finding $u \in L_1$ such that the degree of $u$ is greater than $2k^2$ and replacing it via slides can be done in $O(k)$-time. This procedure will be repeated at most $2k$ times and hence step~(b) requires $O(k^2 + kn^2)$-time.
Going to step (c), we can test isomorphism of components using any exponential-time algorithm. Since the size of the individual components is now bounded by $k^{O(k^3)}$ and the algorithm will run on all pairs of components, step (c) will require $2^{k^{O(k^3)}}$-time in the worst case.
Finally, step~(d) consists only of removing components and can be done in $O(n)$. Therefore step~(1) will take $O(kn^3 + 2^{k^{O(k^3)}})$-time.

\paragraph*{Complexity of step (2).}
Building the graph $\mathcal{R}(G, k)$ will take $O(|V(\mathcal{R}(G, k))| + k^2 |V(\mathcal{R}(G, k))|^2) = O(k^2{f(k) \choose k}^2)$-time since we can check naively for each pair of nodes if they are connected via one slide.

\paragraph*{Complexity of step (3).}
The breadth-first search traversal will take $O(|V(\mathcal{R}(G, k))| + |E(\mathcal{R}(G, k))|) = O({f(k) \choose k}^2)$-time.

\paragraph*{Putting it all together.}
Therefore, the total running time of the algorithm is $$O(kn^3) + 2^{k^{O(k^3)}} + O(k^2{f(k) \choose k}^2)$$ and hence we get the desired result.

\begin{theorem}
{\sc Token Sliding} is fixed-parameter tractable when parameterized by $k$ on graphs of girth five or more.
\end{theorem}